%% file: Fault-InducedDynamics.tex
\documentclass[letterpaper,11pt]{article}

\usepackage[utf8]{inputenc}
\usepackage[letterpaper,margin=1in]{geometry}
\usepackage{amsfonts}
\usepackage{amsmath}
\usepackage{amsthm}
\usepackage{amssymb}
\usepackage{algorithm}
\usepackage[noend]{algpseudocode}
\usepackage[hyphens]{url}
\usepackage{hyperref}
\usepackage{graphicx}
\usepackage{subfigure}
\usepackage{xspace}
\usepackage{units}
\usepackage{color}

\usepackage[compact]{titlesec}
\titlespacing{\section}{0pt}{2ex}{1ex}
\titlespacing{\subsection}{0pt}{0.75ex}{0.5ex}
\titlespacing{\subsubsection}{0pt}{0.5ex}{0ex}
\usepackage{paralist}

\newtheorem*{theorem*}{Theorem}
\newtheorem{theorem}{Theorem}
\newtheorem{corollary}[theorem]{Corollary}

\newtheorem{lemma}[theorem]{Lemma}

\newtheorem{definition}{Definition}

\newcommand{\Circ}{{\cal C}\xspace}
\newcommand{\setrobots}{{\cal R}\xspace}
\newcommand{\segment}[1]{\ensuremath{\overline{#1}}\xspace}

\begin{document}

\title{\bf Fault-Induced Dynamics \\
of Oblivious  Robots on a Line 
\thanks{This work has been supported in part by the
Natural Sciences and Engineering Research Council of Canada
through the Discovery Grant program;
by Prof. Flocchini's University Research Chair.}
}

\author{Jean-Lou De Carufel $^\dagger$,
Paola Flocchini$^\dagger$
\medskip\\
\small
$^\dagger$ University of Ottawa, Canada. {\tt  \{jdecaruf,paola.flocchini\}@uottawa.ca}
} 

\date{}

\maketitle
\thispagestyle{empty}
 \begin{abstract}
 The study of computing in presence of faulty robots in the {\sc Look-Compute-Move} 
 model has been the object of extensive investigation, typically 
 with the goal of designing  algorithms tolerant to as many faults as possible.
 
 In this paper, we initiate a new line of investigation on the presence of  faults,  
 focusing on a rather different issue.
 We are interested in understanding the dynamics of a group of robots  when they execute 
 an algorithm designed for a fault-free environment,  in presence of some undetectable  
 crashed robots.
We start this   investigation focusing on  the classic point-convergence algorithm 
by Ando et al. \cite{AnOaSuYa99}  for robots     with limited visibility, 
in   a  simple  setting (which already presents serious challanges): the 
 robots operate  fully synchronously on a line,  and at most two of them are faulty.
Interestingly, and perhaps surprisingly,   the presence of faults induces the robots to perform some form of {\em scattering}, 
rather than {\em point-convergence}. In fact, we discover that  they arrange  themselves inside the segment 
delimited by the two faults  in   interleaved sequences of equidistant robots. 
The structure that they form has a  hierarchical nature:    robots organize themselves in groups where
a  group  of some level    converges to  an equidistant distribution  
 only  after all groups of lower  levels have converged.

This is the first study on the unintended  dynamics of oblivious robots induced by the presence of  faults.

   \end{abstract}

\section{Introduction}
\input{Introduction.tex}

\section{Preliminaries}
\input{Preliminaries.tex}

\section{Robots' Dynamics in Presence of Two Faults}
\input{TwoFaults.tex}

\section{Conclusion}
\input{Conclusion.tex}

\end{document}

%% file: Introduction.tex

Consider a group of   robots represented as points, which   operate in a continuous space according to 
the   {\sc Look}-{\sc Compute}-{\sc Move} model  \cite{FloPS12}:
when active, a robot {\sc Looks} the environment obtaining a snapshot of the positions of the other robots, 
it {\sc Computes} a destination point on the basis of such a snapshot, and it {\sc Moves} there.
As typically assumed by the model, the robots   are \emph{anonymous} (i.e., they are  identical), 
\emph{autonomous} (without central or external control),
\emph{oblivious} (they have no memory of past activations),
 \emph{disoriented } (they do not agree on a  common  coordinate systems),
  \emph{silent } (they have no means of explicit communication). 
These  systems of autonomous  robots have been extensively investigated 
under different assumptions on the various model parameters (different levels of synchrony, 
level of agreement on the coordinate system, etc.),
and most algorithms in the literature are designed for  fault-free groups of robots  (e.g., see \cite{CieFPS12,DieLP08,DP09,FPSV15,FloPS08,FYOKY15,IzuSK+12,YamS10,YUKY15}).

There are several studies that consider the presence of  faults:
crashes (robots that are never activated) or byzantine (robots that behave differently than intended). The goal, in these cases, has   been to design  fault-tolerant  algorithms focusing on the  maximum amount of faults that  can be tolerated for a solution to exist in a given model (e.g., see  \cite{AP06,ABCTU13,BDT13,BGT10,BT15b,DGMR06}). For a detailed  account of the current investigations see \cite{DGMR06}.

In this paper, we consider a rather different question in presence of faulty robots that has never been asked before. 
Given an algorithm  designed to achieve
a certain global goal by a group of fault-free robots, what is the behaviour of the robots in presence of crash faults ?
Clearly, in most cases, the original  goal is not achieved, but the theoretical  interest is in   characterizing the dynamics of the non-faulty robots induced by the presence of the faulty ones,   from  arbitrary initial configurations. 

We start this new line of investigation  focusing on  the classic point-convergence  algorithm 
by Ando et al. \cite{AnOaSuYa99} for robots with limited visibility, 
and  considering   one of the simplest possible settings, which already proves to be challenging: 
fully synchronous robots ({\sc FSynch})  moving in a 1-dimensional space (a line), in presence of at most two faults. 
%
%
  In a line, the convergence  algorithm prescribes each robot to move to the center of the two farthest visible robots and, in absence of faults, starting from a configuration where the robots'  ``visibility graph" is connected, the robots are guaranteed to converge toward a point.  It is not difficult to see that with a single fault in the system,   the robots successfully converge toward the faulty robot. The presence of multiple faults, however,  gives rise to intricate dynamics, and the  analysis of the robots behaviour  is  already quite complex with just two  faults.

Interestingly, and perhaps surprisingly,   the presence of faults induces the robots to perform some form of {\em scattering}, 
rather than {\em gathering}. In fact, we discover that they arrange  themselves inside the segment 
delimited by the two faults  in   interleaved sequences composed of equidistant robots. 
The structure that they form has a  hierarchical nature:    robots organize themselves in groups where
a  group  of some level     converges to  an equidistant distribution  between the first and the last robots 
of that group,  only  after all groups of lower  levels have converged.
Also interesting to note is the rather different dynamics that arises   
when moving  to the middle between two  robots, depending on the choice of the robots:
when considering the  {\em closest}  neighbours,  the result is an equidistant distribution  
(scattering algorithm of  \cite{DP09}),  when instead selecting the   {\em farthest}  robots   the result is a 
more complex structure of sequences of robots, each converging to an  equidistant  distribution.
The main difficulty of our analysis is to show that the robots indeed  form  this  special combination 
of   sequences:  the convergence of each sequence is then derived from  a generalization of the result by \cite{DP09}.

Finally observe that the  2-dimensional case  has a rather different nature. In fact, in contrast to the 1-dimensional setting, where any initial configuration converges toward a pattern, when robots move on the plane  oscillations are possible, even with just two faults. The investigation  of this case is left for future study.

%% file: Preliminaries.tex

\subsection{Model and Notation}

Let $X$ denote a set of identical robots moving on a line,     simultaneously activated in synchronous time steps according to the {\sc Look}-{\sc Compute}-{\sc Move} model.  At each activation, the robots ``see" the positions of the ones that are visible to them   (each  robot  can see up to a fixed distance $V$), they all compute a destination point, and they move  to that point.
The robots are oblivious in the sense that the computation at time $t$ solely depends on the positions of the robots perceived at that step.
We assume that two  robots are permanently faulty and  do not participate in any activity; their faulty status, however, is not visible  and they appear identical to the others. 

Let $x$ denote robot a $x \in X$ and $x(t)$ its position at time $t$ with respect to the leftmost faulty robot, 
which is considered to be at position $0$ on the real axis.
Let $X(t)=\{x_0(t), x_1(t), ..., x_n(t)\}$ denote the {\em configuration} of robots at time $t$.
 With an abuse of notation $x(t)$ may indicate both the robot itself and its position at time $t$.
 Robots do not necessarily occupy distinct positions. For instance we might have $x_i(t) = x_j(t)$ where $0\leq i,j\leq n$ are two different indices. 
Note however that non-faulty robots in the same position behave  in the same way  and can be considered as a single one. 
Indeed, when non-faulty  robots end up in the same position, we say that they  ``merge"
 and from that moment on they will be considered as one.

Throughout the paper, we suppose that $x_0$ is the leftmost faulty robot and $x_n$ is the rightmost faulty robot. Therefore,
for all $t\geq 0$, we have $x_0(t) = 0$ and $x_n(t)$ is equal to some positive fixed position on the real axis.
 
We denote the distance between robots $x$ and $y$ at time $t$
by $|x(t)-y(t)|$.  
We denote by $[\alpha,\beta]$ the interval of real numbers starting at $\alpha\in\mathbb{R}$ and ending at $\beta\in\mathbb{R}$, where $\alpha \leq \beta$.
Let $N(x(t))$ be the set of robots visible by $x$ at time $t$.
In other words, for all $y \in N(x(t))$, $|x(t)-y(t)| \leq V$.
Let $r(x(t))$ (respectively $l(x(t))$)
denote the rightmost (respectively the leftmost)
visible robot from robot $x$ at time $t$.

 We say that a configuration of robots $X=\{x_0, x_1, ..., x_n\}$  converges to a pattern $P = \{p_0, p_1, ..., p_n\}$ if
for all $0\leq i \leq n$, $x_i(t)\rightarrow p_i$ as $t\rightarrow\infty$.

\subsection{Point-Convergence}

A classical problem for oblivious robots is  gathering: the robots, initially placed in arbitrary positions, need to find themselves on the same   point, not  established a-priori.
The convergence version of the problem requires the robots to converge toward a point. A solution to this problem is given by the well known   algorithm by Ando et al. \cite{AnOaSuYa99}.  The algorithm achieves convergence to a point, not only in synchronous systems, but  also when at each time step, only a subset of the robots is activated (semi-synchronous scheduler  {\sc SSynch})), as long as every robot is activated infinitely often.

\begin{center}

\fbox{
\begin{minipage}{12cm}
{\sc  Convergence2D} \cite{AnOaSuYa99} (for robot $x_i$) \\
\vspace{-0.5cm}
\begin{itemize}
\item  $\forall x_j \in N(x_i) \setminus\{x_i\}$,
\begin{itemize}
\item  $d_j := dist(x_i, x_j)$,
\item $\theta_j :=c_i \widehat{x_i} x_j$,
\item  $l_j :=(d_j/2)\cos \theta_j + \sqrt{(V/2)^2-((d_j/2)\sin
\theta_j)^2}$,
\end{itemize}
\item $limit := \min_{x_j \in \setrobots_i(t)\setminus\{x_i\}} \{l_j\}$,
\item  $goal := dist(x_i, c_i)$,
\item $D := \min \{goal,limit\}$,
\item $p :=$ point on $\segment{x_ic_i}$ at distance $D$ from $x_i$.
\item Move towards  $p$.
\end{itemize}
\end{minipage}
}

 \end{center}

Robots are initially placed in arbitrary positions in a 2-dimensional space,  with  limited visibility.
Let $SC_i(t)$ denotes 
the smallest enclosing circle  of  
the  positions of   robots
in $ {\cal R}(t)$ seen by $x_i$; let $c_i(t)$ be the center of $SC_i(t)$.  
According to the algorithm,   $x_i$   moves
toward $c_i(t)$, but only up to a certain distance. 
Specifically, 
 its
 destination  is  the point   on the segment
$\overline{{x_i(t)}{c_i(t)}}$ that is closest to $c_i(t)$ and that
satisfies the following condition:
For every robot $x_j \in N(x_i(t))$, $p$ lies in the disk $\Circ_i$
whose center is the midpoint    of $x_i(t)$ and $x_j(t)$, and whose
radius is $V/2$.
This condition ensures that $x_i$ and $x_j$ will
still be visible after the movement of  $x_i$, and possibly of $x_j$.  

%


\paragraph {\bf The 1-Dimensional Case.}
Consider now the same algorithm in the  particular case of a   one-dimensional setting where the space where the robots can move is a line.
In this setting, the algorithm ({\sc Convergence1D}) becomes quite simple because the smallest enclosing circle of the visible robots is the segment delimited by the two farthest apart robot, and a robots moves to occupy the mid-point between those two robots.

\begin{center}
\fbox{
\begin{minipage}{12.5cm}
{\sc Convergence1D} (for robot $x_i$ at time step $t$)
\begin{itemize}
\item Let $l(x_i(t))$ and $r(x_i(t))$ be the farthest visible robots.
\item  move to  the midpoint between $l(x_i(t))$ and $r(x_i(t))$.
\end{itemize}
\end{minipage}
}
\end{center}

\begin{theorem}\cite{AnOaSuYa99}\
\label{thm:Ando}
Executing Algorithm {\sc Convergence1D} in {\sc FSynch} or  {\sc SSynch}, the robots converge to a point.
\end{theorem}

\subsection {Spreading on a line}
In~\cite{CoPe08}, a classical spreading algorithm for robots in 1-dimensional systems has been analyzed both in {\sc FSynch} and {\sc SSynch}. A variant of this result, which is derived below, will be used in the   paper when proving convergence to a pattern.
 
Consider a set of oblivious robots $X = \{x_0, x_1, ..., x_n\}$ on a line  that follow the {\sc Look-Compute-Move} model, where $x_0$ and $ x_n$ do not move (equivalently, this can be considered as a segment delimited by the positions of $x_0$ and $x_n$).  
Let $|x_0(0),x_n(0)|= D$. The robots have {\em neighbouring visibility}, which means that they are able to see the two closest robots 
(while $x_0$ and $x_n$ know they are the delimiters of the segment).
The algorithm of \cite{CoPe08} ({\sc Spreading}) makes the robot converge to a configuration where the distance  between consecutive robots tends to $\frac{D}{n}$ by having the extremal robots never move and the others  move to the middle point between the two neighbouring robots. 

\begin{center}
\fbox{
\begin{minipage}{12.5cm}
{\sc Spreading} (for robot $x_i$ at time step $t$)
\begin{itemize}
\item If I am an extremal robot:  do nothing.
\item Let $x_i(t)^-$ and $x_i(t)^+$ be the closest visible robots.
  \item  move to  the midpoint between  $x_i(t)^-$ and $x_i(t)^+$.
\end{itemize}
\end{minipage}
}
\end{center}

\begin{theorem}\cite{CoPe08}
\label{thm:spreading}
Executing  Algorithm {\sc Spreading} in {\sc FSynch} or in  {\sc SSynch} on the  set of robots $R$ where the first and the last robots do not move, 
the robots converge to equidistant positions.
\end{theorem}

We now prove that, in {\sc FSynch}, convergence is achieved using the same algorithm  also in a slightly different setting. In fact, we consider the case when $x_0$ and $x_{n}$ are not still, but they are each converging towards a point. The generalization to this case is not straightforward. 

\begin{theorem}
\label{thm:generalizedCohen}
Let $X = \{x_0, x_1, ..., x_n\}$ where  
$x_0(t) \rightarrow x_0'$ and  $x_n(t) \rightarrow x_n'$ as $t \rightarrow  \infty$.
Executing  Algorithm {\sc Spreading} in {\sc FSynch} on the set of robots $\{x_1, ..., x_{n-1}\}$, 
the robots converge to equidistant positions between $x_0'$ and $x_n'$ .
\end{theorem}

\begin{proof}
Without loss of generality,
suppose that $x_0' = 0$ and $x_n' = 1$.
We want to prove that $x_i(t) \rightarrow \frac{i}{n}$ as $t\rightarrow \infty$ for all $1\leq i \leq n-1$.
We follow the proof of Theorem~\ref{thm:spreading}.
For all $1\leq i \leq n-1$,
the next position of $x_i(t)$ is
$$x_i(t+1) = \frac{x_{i-1}(t) + x_{i+1}(t)}{2} .$$
Let $\eta_i(t) = x_i(t) - \frac{i}{n}$
for all $0\leq i \leq n$.
We get
$$\eta_i(t+1) = \frac{\eta_{i-1}(t) + \eta_{i+1}(t)}{2}$$
for all $1 \leq i \leq n-1$.
Our goal is to show that $\eta_i(t) \rightarrow 0$ for all $0 \leq i \leq n$. By the hypothesis, we already know that
$\eta_0(t) \rightarrow 0$ and $\eta_n(t) \rightarrow 0$ as $t\rightarrow \infty$.
The fact that $\eta_i(t) \rightarrow 0$ as $t\rightarrow \infty$ for all $1\leq i \leq n-1$ relies on the following lemma.
\end{proof}

The following lemma is a generalization of a result by Cohen and Peleg~\cite{CoPe08}.
\begin{lemma}
Let $\eta_i(t)$ be a sequence of real numbers for all $0 \leq i \leq m$, where $m\geq 2$ is an integer. Suppose that
$$\eta_i(t+1) \leq \frac{\eta_{i-1}(t) + \eta_{i+1}(t)}{2}$$
for all $1 \leq i \leq m-1$ and for all $t \geq 0$,
and $\eta_0(t) \rightarrow 0$ and $\eta_m(t) \rightarrow 0$
as $t\rightarrow \infty$.
Moreover,
suppose that there exists a positive real number $M$
such that $|\eta_i(t)| \leq M$
for all $0 \leq i \leq m$ and for all $t \geq 0$.
Then,
for all $0 \leq i \leq m$,
$\eta_i(t) \rightarrow 0$ as $t\rightarrow \infty$.
\end{lemma}

\proof
We show that for all $0 \leq i \leq m$,
$\eta_i(t) \rightarrow 0$ as $t\rightarrow \infty$.
By definition, this is true for $i \in \{0,m\}$.
To deal with other values of $i$,
let
$$\psi(t) = \sum_{i=0}^m\eta^2_i(t) .$$
We show that $\psi(t)\rightarrow 0$ as $t\rightarrow \infty$, which completes the proof.
Following the same approach as the one used in the proof of Theorem~\ref{thm:spreading}, we use the Fourier sine series of $\eta_i(t)$.
However, in our case,
we need to be careful since $\eta_0(t)$ and $\eta_m(t)$ are not necessarily equal to $0$.
Let
\begin{align*}
g(i,k) &= \sqrt{\frac{2}{m}}\sin\left(\frac{ki\pi}{m}\right) ,\\
\mu_k &= \sum_{i=0}^m \eta_i g(i,k) .
\end{align*}
We have
$$\eta_i = \sum_{k=0}^m\mu_k g(i,k)$$
for all $1 \leq i \leq m-1$.
Moreover,
we have
$$g(0,k) = g(m,k) = 0 $$
for all $0 \leq k \leq m$,
from which
$$\sum_{k=0}^m\mu_k g(0,k) = \sum_{k=0}^m\mu_k g(m,k) = 0.$$
Notice that
$$\sum_{i=0}^m g(i,k)g(i,q) = \delta_{k,q},$$
where $\delta_{k,q}$ stands for the Kronecker's delta, i.e., $\delta_{k,q} = 1$ if $k = q$ and $0$ otherwise.
Moreover, observe that for all $0 \leq k \leq m$,
$$|\mu_k| = \left|\sum_{i=0}^m \eta_i g(i,k)\right| = \left|\sum_{i=1}^{m-1} \eta_i g(i,k)\right| \leq \sum_{i=1}^{m-1} \left|\eta_i g(i,k)\right| \leq (m-1) \sqrt{\frac{2}{m}} M \leq \sqrt{2m}M.$$

Since $\eta_0(t) \rightarrow 0$ and $\eta_m(t) \rightarrow 0$ as $t\rightarrow \infty$,
for all $\epsilon > 0$,
there is a time $t_0 \geq 0$ such that
$|\eta_0(t)| < \frac{\sqrt{2}}{2}\epsilon$ and $|\eta_m(t)| < \frac{\sqrt{2}}{2}\epsilon$
for all $t\geq t_0$.
Therefore,
for all $t\geq t_0$,
we have
\begin{align*}
\psi(t) &= \sum_{i=0}^m\eta^2_i(t) \\
&\leq \epsilon^2 + \sum_{i=1}^{m-1}\eta^2_i(t) \\
&= \epsilon^2 + \sum_{i=1}^{m-1}\left(\sqrt{\frac{2}{m}}\sum_{k=0}^m\mu_k(t)\sin\left(\frac{ki\pi}{m}\right) \right)^2 \\
&= \epsilon^2 + \sum_{i=0}^m\left(\sqrt{\frac{2}{m}}\sum_{k=0}^{m}\mu_k(t)\sin\left(\frac{ki\pi}{m}\right) \right)^2 \\
&= \epsilon^2 + \sum_{i=0}^m\left(\sum_{k=0}^m\sqrt{\frac{2}{m}}\mu_k(t)\sin\left(\frac{ki\pi}{m}\right) \right)^2 \\
&= \epsilon^2 + \sum_{i=0}^m\sum_{k=0}^m\sum_{q=0}^m
\left(\sqrt{\frac{2}{m}}\mu_k(t)\sin\left(\frac{ki\pi}{m}\right)\right)
\left(\sqrt{\frac{2}{m}}\mu_q(t)\sin\left(\frac{qi\pi}{m}\right)\right) \\
&= \epsilon^2 + \sum_{k=0}^{m}\sum_{q=0}^{m}\mu_k(t)\mu_q(t)\sum_{i=0}^m \left(\sqrt{\frac{2}{m}}\sin\left(\frac{ki\pi}{m}\right)\right)
\left(\sqrt{\frac{2}{m}}\sin\left(\frac{qi\pi}{m}\right)\right)\\
&= \epsilon^2 + \sum_{k=0}^m\sum_{q=0}^m\mu_k(t)\mu_q(t)\sum_{i=0}^m g(i,k)g(i,q)\\
&= \epsilon^2 + \sum_{k=0}^m\sum_{q=0}^m\mu_k(t)\mu_q(t)\delta_{k,q}\\
&= \epsilon^2 + \sum_{k=0}^m \mu_k^2(t) .
\end{align*}
Moreover,
for $1 \leq k  \leq m-1$
(since by definition, $\mu_0(t+1) = \mu_m(t+1) = 0$),
we have
\begin{align*}
&\quad \mu_k(t+1) \\
=&\quad \sum_{i = 0}^m \sqrt{\frac{2}{m}}\sin\left(\frac{ki\pi}{m}\right)\eta_i(t+1) \\
=&\quad \sum_{i = 1}^{m-1} \sqrt{\frac{2}{m}}\sin\left(\frac{ki\pi}{m}\right)\eta_i(t+1) \\
\leq&\quad \sum_{i = 1}^{m-1} \sqrt{\frac{2}{m}}\sin\left(\frac{ki\pi}{m}\right)\frac{\eta_{i-1}(t)+\eta_{i+1}(t)}{2} \\
=&\quad \sum_{i = 0}^{m-2} \frac{1}{2}\sqrt{\frac{2}{m}}\sin\left(\frac{k(i+1)\pi}{m}\right) \eta_i(t) 
+ \sum_{i = 2}^{m} \frac{1}{2}\sqrt{\frac{2}{m}}\sin\left(\frac{k(i-1)\pi}{m}\right) \eta_i(t)\\
=&\quad \sum_{i = 0}^{m-1} \frac{1}{2}\sqrt{\frac{2}{m}}\sin\left(\frac{k(i+1)\pi}{m}\right) \eta_i(t) 
+ \sum_{i = 1}^{m} \frac{1}{2}\sqrt{\frac{2}{m}}\sin\left(\frac{k(i-1)\pi}{m}\right) \eta_i(t)\\
\leq&\quad \epsilon + \sum_{i = 0}^m \frac{1}{2}\sqrt{\frac{2}{m}}\sin\left(\frac{k(i+1)\pi}{m}\right) \eta_i(t) 
+ \sum_{i = 0}^m \frac{1}{2}\sqrt{\frac{2}{m}}\sin\left(\frac{k(i-1)\pi}{m}\right) \eta_i(t)\\
=&\quad \epsilon + \frac{1}{2}\sqrt{\frac{2}{m}}\sum_{i = 0}^m \left(\sin\left(\frac{k(i+1)\pi}{m}\right)  + \sin\left(\frac{k(i-1)\pi}{m}\right)\right)\eta_i(t)\\
=&\quad \epsilon + \sqrt{\frac{2}{m}}\sum_{i = 0}^{m} \sin\left(\frac{ki\pi}{m}\right)\cos\left(\frac{k\pi}{m}\right)\eta_i(t)\\
=&\quad \epsilon + \cos\left(\frac{k\pi}{m}\right)\mu_k(t)
\end{align*}

Therefore,
\begin{align*}
\psi(t+1) &\leq \epsilon^2 + \sum_{k=0}^m\mu^2_k(t+1) \\
&= \epsilon^2 + \sum_{k=1}^{m-1}\mu^2_k(t+1) \\
&\leq \epsilon^2 + \sum_{k=1}^{m-1}\left(\epsilon + \cos\left(\frac{k\pi}{m}\right)\mu_k(t)\right)^2 \\
&= \epsilon^2 + \sum_{k=1}^{m-1}\left(\epsilon^2 + 2\epsilon\cos\left(\frac{k\pi}{m}\right)\mu_k(t) + \cos^2\left(\frac{k\pi}{m}\right)\mu_k^2(t)\right) \\
&\leq m\epsilon^2 + 2\epsilon (m-1)\cos\left(\frac{\pi}{m}\right)\sqrt{2m}M + \cos^2\left(\frac{\pi}{m}\right)\psi(t) \\
&= \left(m\epsilon + 2(m-1)\cos\left(\frac{\pi}{m}\right)\sqrt{2m}M\right)\epsilon + \cos^2\left(\frac{\pi}{m}\right)\psi(t) \\
&= \Phi\epsilon + \Upsilon\psi(t),
\end{align*}
where
$$\Phi = m\epsilon + 2(m-1)\cos\left(\frac{\pi}{m}\right)\sqrt{2m}M$$
is bounded above by a constant and
$$\Upsilon = \cos^2\left(\frac{\pi}{m}\right) < 1 .$$

Consequently, we have
\begin{align*}
\psi(t_0+t') &\leq \Phi\epsilon\left(1+\Upsilon+\Upsilon^2+...+\Upsilon^{t'-1}\right)+\Upsilon^{t'}\psi(t_0) \\
&\leq \Phi\epsilon\frac{1}{1-\Upsilon}+\Upsilon^{t'}\psi(t_0) 
\end{align*}
and for $t'$ sufficiently large,
we get
$$\Phi\epsilon\frac{1}{1-\Upsilon}+\Upsilon^{t'}\psi(t_0) \leq \Phi\epsilon\frac{1}{1-\Upsilon}+\epsilon = \left(\Phi\frac{1}{1-\Upsilon}+1\right)\epsilon ,$$
which can be made arbitrarily small by an appropriate choice of $\epsilon$.
\qed

%% file: TwoFaults.tex

It is easy to see that if the configuration contains a single faulty robots, the other  robots converge toward it.
In this Section we then  focus on the case when the system contains two faults and we show that,  starting from an arbitrary  configuration  
and following algorithm  {\sc Convergence1D}, the system converges towards a limit configuration. 
%
Intuitively, we will show the convergence of $X$ by showing that the robots will eventually form
a ``hierarchical" structure of ``independent"  groups, 
where the robots at a certain level of the hierarchy 
move $\epsilon$-close to their convergence point  
after the lower levels robots have already done so.

\subsection{Basic Properties} 
We start with a series of lemmas leading to the proof of two crucial properties:  there exists a time after which robots preserve their farthest neighbours (Theorem \ref{thm:preserved-farthest}) and  there exists a time after which the number of different positions  occupied by them becomes constant (Corollary \ref{cor:stable}). 

\begin{lemma}[No Crossing]
\label{lem:no-crossing}
If $x$ and $z$ are two non-faulty robots and $x(t) < z(t)$, then $x(t+1) \leq z(t+1)$.
\end{lemma}

\begin{proof}
Since $x(t) < z(t)$,
we have that $r(x(t)) \leq r(z(t))$ and $l(x(t)) \leq l(z(t))$
by definition.
It follows that $x(t+1) = \frac{l(x(t))+ r(x(t))}{2} \leq  \frac{l(z(t))+ r(z(t))}{2} = z(t+1)$.
\end{proof}

With the next three lemmas
(\ref{lem:segment}, \ref{lem:no-more-crossing}, and~\ref{lem:at most two outsiders}),
we show that all robots, 
except possibly two,
eventually enter the segment  $[x_0,x_n]$ delimited by the two faulty robots. 
At most two robots might perpetually stay outside of it, one to the left of $x_0$ and one to the right of $x_n$. If this is the case, however, the two outsiders converge to $x_0$ and $x_n$, respectively.

\begin{lemma}
\label{lem:segment}
Either one of the following two scenarios happens as $t\rightarrow \infty$.
\begin{enumerate}
\item In a finite number of steps,
all robots place themselves
inside the line segment $[x_0,x_n]$.

\item There is at least one robot $x$ that never enters the line segment $[x_0,x_n]$. If $x(0) < x_0$,
then $x(t)$ tends towards $x_0$ as $t\rightarrow \infty$.
If $x(0) > x_n$,
then $x(t)$ tends towards $x_n$ as $t\rightarrow \infty$.
\end{enumerate}
\end{lemma}

\begin{proof}
Since the two faulty robots do not move, they are already inside
$[x_0,x_n]$. For the rest of the proof, we consider only the non-faulty robots. Let $x_{\ell}$ and $x_r$ be the leftmost and the rightmost non-faulty robots, respectively. 
By Lemma~\ref{lem:no-crossing},
$x_{\ell}$ (respectively $x_r$)
stays the leftmost
(respectively the rightmost)
non-faulty robot at all steps of the execution of the algorithm.
Therefore,
it is sufficient to prove the lemma for $x_{\ell}$ and $x_r$.

We first argue that if at some time $t_0 > 0$,
$x_{\ell}(t_0) \in [x_0,x_n]$, then for all $t > t_0$,
$x_{\ell}(t) \geq x_0$.
Since $x_{\ell}$ is the leftmost non-faulty robot
and $x_{\ell}(t_0) \in [x_0,x_n]$,
we have $l(x_{\ell}(t_0)) \geq x_0$.
Therefore,
$$x_{\ell}(t_0+1) = \frac{l(x_{\ell}(t_0))+r(x_{\ell}(t_0))}{2} \geq \frac{x_0+x_{\ell}(t_0)}{2} \geq \frac{x_0+x_0}{2} = x_0 ,$$
from which the proof follows by induction on $t$.
A symmetric argument shows that if $x_r(t_0) \in [x_0,x_n]$, then for all $t > t_0$,
$x_r(t) \leq x_n$.
It remains to consider the case
where $x_{\ell}$ or $x_r$ never enters $[x_0,x_n]$.

Suppose that $x_{\ell}$ does not enter the interval $[x_0,x_n]$ in a finite number of steps. Therefore\footnote{The case where the leftmost robot is to the right of $x_n$ is taken care of by the case where the rightmost robot is to the right of $x_n$.},
$x_{\ell}(t) < x_0$ for all $t \geq 0$.
Together with the fact that
$x_{\ell}$ is the leftmost non-faulty robot,
we get $l(x_{\ell}(t))=x_{\ell}(t)$ and $r(x_{\ell}(t)) > x_{\ell}(t)$ for all $t\geq 0$. Therefore,
$$x_{\ell}(t + 1) = \frac{l(x_{\ell}(t))+r(x_{\ell}(t))}{2}
> \frac{x_{\ell}(t)+x_{\ell}(t)}{2} = x_{\ell}(t) $$
for all $t \geq 0$.
Thus, it follows that $x_{\ell}(t)$ is strictly increasing for $t \geq 0$. Since $x_{\ell}$ never enters the interval $[x_0,x_n]$,
$x_{\ell}(t) < x_0$ for all $t \geq 0$. Therefore,
$x_{\ell}(t)$ converges to a point $x_{\ell}^* \leq x_0$ as $t\rightarrow \infty$.

We show that $x_{\ell}^* = x_0$ by contradiction.
Suppose that $x_{\ell}^* < x_0$.
Since $x_{\ell}(t)$ is strictly increasing for $t \geq 0$
and $x_{\ell}(t)$ converges to $x_{\ell}^*$ as $t\rightarrow \infty$,
then $x_{\ell}(t) < x_{\ell}^*$ for all $t \geq 0$.
Let $t_0 \geq 0$ be a time
for which $x_{\ell}^* - x_{\ell}(t_0) = \delta < \frac{V}{4}$.
Let $x_1(t_0) = r(x_{\ell}(t_0))$
and $\delta' = x_1(t_0) - x_{\ell}^*$.
We do not know whether $x_1(t_0)$ is to the left or to the right of $x_{\ell}^*$.
In other words, we do not know the sign of $\delta'$.
Since $x_{\ell}(t)$ is strictly increasing
and $x_{\ell}(t) < x_{\ell}^*$ for all $t \geq 0$,
we have $|\delta'| < \delta$.
Therefore,
$x_{\ell}^* - x_{\ell}(t_0+1) = \frac{\delta-\delta'}{2}$.
We now look at the rightmost visible robot from $x_1(t_0)$.
We have $r(x_1(t_0)) - x_{\ell}(t_0) > V$,
otherwise $x_1(t_0)$ would not be the rightmost visible robot from $x_{\ell}(t_0)$.
Therefore,
we have
\begin{align}
\label{ineq rx1 x1}
r(x_1(t_0)) - x_1(t_0) = (r(x_1(t_0)) - x_{\ell}(t_0)) + (x_{\ell}(t_0) - x_1(t_0)) > V - (\delta + \delta') .
\end{align}
We also have
\begin{align}
\label{ineq x1 xl t0+1}
x_1(t_0+1) - x_{\ell}(t_0+1)
&= \frac{l(x_1(t_0)) + r(x_1(t_0))}{2} - \frac{l(x_{\ell}(t_0)) + r(x_{\ell}(t_0))}{2} \\
&= \frac{x_{\ell}(t_0) + r(x_1(t_0))}{2} - \frac{x_{\ell}(t_0) + x_1(t_0)}{2} \\
\label{ineq x1 xl t0+1 inside}
&= \frac{r(x_1(t_0)) - x_1(t_0)}{2} \\
&< V ,
\end{align}
from which $x_1(t_0+1)$ is visible from $x_{\ell}(t_0+1)$.
This leads to
\begin{align*}
x_{\ell}(t_0+2) - x_{\ell}^*
&= \frac{l(x_{\ell}(t_0+1))+r(x_{\ell}(t_0+1))}{2}  - x_{\ell}^* \\
&\geq \frac{x_{\ell}(t_0+1)+x_1(t_0+1)}{2}  - x_{\ell}^* \\
&= \frac{(x_{\ell}(t_0+1) - x_{\ell}^*) + (x_1(t_0+1) - x_{\ell}^*)}{2} \\
&= \frac{(x_{\ell}(t_0+1) - x_{\ell}^*) + (x_1(t_0+1) - x_{\ell}(t_0+1)) + (x_{\ell}(t_0+1) - x_{\ell}^*)}{2} \\
&= \frac{2(x_{\ell}(t_0+1) - x_{\ell}^*) + (x_1(t_0+1) - (x_{\ell}(t_0+1))}{2} \\
&> \frac{(\delta'-\delta) + \frac{V-(\delta+\delta')}{2}}{2} & \text{from~\eqref{ineq x1 xl t0+1},
\eqref{ineq x1 xl t0+1 inside}
and~\eqref{ineq rx1 x1}}\\
&= \frac{V-3\delta+\delta'}{4} \\
&> \frac{V-4\delta}{4} \\
&> 0,
\end{align*}
from which $x_{\ell}(t_0+2) > x_{\ell}^*$,
which is a contradiction
since $x_{\ell}(t) < x_{\ell}^*$ for all $t \geq 0$.

Since $x_{\ell}(t)$ converges to $x_0$ as $t\rightarrow \infty$,
$x_{\ell}$ will be at distance at most $\epsilon$ from $x_0$ in a finite number of steps.

A symmetric argument completes the proof for $x_r$.
\end{proof}

\begin{lemma}[No More Crossing]
\label{lem:no-more-crossing}
If $x$ is a non-faulty robot,
it will cross at most a finite number of times with a faulty robot.
\end{lemma}

\begin{proof}
Let $x_{\ell}$ be the leftmost non-faulty robot.
From the proof of Lemma~\ref{lem:segment}, 
we know that $x_{\ell}$ will stay the leftmost
non-faulty robot at all steps of the execution of the algorithm.
Moreover,
from Lemma~\ref{lem:segment},
two scenarios are possible:
(1) $x_{\ell}$ enters the line segment $[x_0,x_n]$ after some time $t_0$ and for all $t \geq t_0$,
$x_{\ell}(t) \geq x_0$
or (2) $x_{\ell}(t)$ is strictly increasing for $t \geq 0$
and $x_{\ell}(t)$ converges to $0$ as $t\rightarrow \infty$.
\begin{enumerate}
\item In this case,
after time $t_0$, no robot will cross $x_0$.

\item In this case,
let $x$ be a robot and $t' > t_0$ be a time such that $x(t') < 0$, $x(t'+1) > 0$ and $x_0 - x_{\ell}(t') = 0 - x_{\ell}(t') = \delta < \frac{V}{2}$. Suppose that there is a time $t'' > t'$ such that
$x(t'') \geq 0$ and $x(t''+1) < 0$.
Since $x_{\ell}$ is the leftmost agent,
we have $l(x(t'')) \geq x_{\ell}(t'')$.
Moreover,
since $x(t''+1) < x_0 = 0$
and $x_{\ell}(t)$ is strictly increasing for $t \geq 0$,
we have
$$r(x(t'')) < x_0 - l(x(t'')) \leq x_0-x_{\ell}(t'') < x_0-x_{\ell}(t') = \delta < V/2 $$
and hence,
$$r(x(t'')) - x_{\ell}(t'') < r(x(t'')) - x_{\ell}(t') < V/2+V/2 = V .$$
Thus,
$l(x(t'')) = x_{\ell}(t'') = l(x_{\ell}(t''))$
and $r(x(t'')) = r(x_{\ell}(t''))$.
Consequently,
$x(t''+1) = x_{\ell}(t''+1)$.
In other words,
$x$ and $x_{\ell}$ merge.
Since $x_{\ell}(t)$ is strictly increasing for $t \geq 0$
and $x_{\ell}(t)$ converges to $0$ as $t\rightarrow \infty$,
$x$ will not cross $0$ anymore.

\end{enumerate}
A symmetric argument with the rightmost non-faulty robot $x_r$ completes the proof.
\end{proof}

\begin{lemma}
\label{lem:at most two outsiders}
Suppose that there is at least one robot $x$ that never enters the line segment $[x_0,x_n]$.
\begin{itemize}
\item If $x < x_0$,
then after a finite number of steps,
$x$ is the only robot on the left of $x_0$.

\item If $x > x_n$,
then after a finite number of steps,
$x$ is the only robot on the right of $x_n$.
\end{itemize}
\end{lemma}

\begin{proof}
We prove this lemma by contradiction.
Suppose that for all $t \geq 0$,
there is a robot $x'$ outside of the line segment $[x_0,x_n]$.

Suppose that $x < x_0$
(the case where $x > x_n$ is symmetric).
Let $x'$ be a robot and $t$ be a time such that
\begin{itemize}
\item $x(t) < x'(t) < x_0(t)$,
\item $\epsilon = x_0 - x(t) < \frac{V}{2}$,
\item and if the rightmost non-faulty robot $x_r$ satisfies $x_r(t) > x_n$ and $x_r$ never enters the line segment $[x_0,x_n]$, then $x_r(t) - x_n < \frac{V}{2}$.
\end{itemize}
We consider two cases: (1) $x'$ eventually enters the line segment $[x_0,x_n]$
or (2) not.
\begin{enumerate}
\item If $x'$ enters the line segment $[x_0,x_n]$ and stays there,
then after a finite number of steps, it is not outside of the line segment $[x_0,x_n]$.

Therefore,
let us consider the case where $x'$ enters the line segment $[x_0,x_n]$ and it eventually gets out of $[x_0,x_n]$.
If $x'$ gets out of $[x_0,x_n]$ by crossing $x_0$,
then it merges with $x$
(refer to the proof of Lemma~\ref{lem:no-more-crossing}).
If $x'$ gets out of $[x_0,x_n]$ by crossing $x_n$,
then it merges with $x_r$
(refer to the proof of Lemma~\ref{lem:no-more-crossing}).

\item If $x'$ never enters the line segment $[x_0,x_n]$,
let $\delta = x_0 - x'(t)$.
Notice that $\delta < \epsilon < \frac{V}{2}$.
Since $x'$ never enters the line segment $[x_0,x_n]$,
$r(x'(t)) - x_0 < \delta < \frac{V}{2}$.
Therefore,
$r(x(t)) = r(x'(t))$
and $x(t) = l(x(t)) = l(x'(t))$.
Therefore,
$x(t+1) = \frac{l(x(t))+r(x(t))}{2} = \frac{l(x'(t))+r(x'(t))}{2} = x'(t+1)$.
\end{enumerate}
In all cases,
if there is a robot $x'$ between $x$ and $x_0$,
then after a finite number of steps,
$x'$ enters the line segment $[x_0,x_n]$ and stays there
or $x'$ merges with another robot.
Since we have a finite number of robots,
after a finite number of steps,
$x$ will be the only robot satisfying $x < x_0$.
\end{proof}

We now show that during the evolution of the system,
a robot never loses visibility of the robots seen in the past.
\begin{lemma}[Preserved Visibility]
\label{lem:preserved-visibility}
Let $y \in N(x(t))$. For all $t'>t$, $y \in N(x(t'))$.
\end{lemma}

\begin{proof}
Let $y \in N(x(t))$. Hence, we have $|x(t) - y(t)| \leq V$.
Without loss of generality, suppose that $y(t)$ is to the left of $x(t)$, from which $0 < x(t) - y(t) \leq V$.
We consider three cases: (1) $x$ and $y$ are non-faulty, (2) exactly one of $x$ and $y$ is faulty, or (3) both $x$ and $y$ are faulty.
\begin{enumerate}
\item In this case, by Lemma~\ref{lem:no-crossing},
$x(t+1) - y(t+1) \geq 0$.
We have
\begin{align*}
x(t+1) - y(t+1) &= \frac{l(x(t)) + r(x(t))}{2} - \frac{l(y(t)) + r(y(t))}{2} \\
&\leq \frac{y(t) + (x(t) + V)}{2} - \frac{(y(t) - V) + x(t)}{2} \\
&= V .
\end{align*}
\item Without loss of generality, suppose that $x$ is faulty and $y$ is non-faulty. If $x(t+1) - y(t+1) \geq 0$,
we have
\begin{align*}
x(t+1) - y(t+1) &= x(t) - \frac{l(y(t)) + r(y(t))}{2} \\
&\leq x(t) - \frac{(y(t) - V) + x(t)}{2} \\
&= \frac{x(t) - y(t) + V}{2} \\
&\leq \frac{V + V}{2} \\
&= V .
\end{align*}

If $y(t+1) - x(t+1) \geq 0$,
we have
\begin{align*}
y(t+1) - x(t+1) &= \frac{l(y(t)) + r(y(t))}{2} - x(t) \\
&\leq \frac{x(t) + (y(t) + V)}{2} - x(t)\\
&= \frac{y(t) - x(t) + V}{2} \\
&\leq \frac{0 + V}{2} \\
&< V .
\end{align*}

\item In this case, we have $x(t+1) = x(t)$ and $y(t+1) = y(t)$, so the result follows.
\end{enumerate}
\end{proof}

A robot never loses visibility of the robots seen in the past; 
however, notice that new robots could enter its visibility range (\emph{inclusion}). 
It is also possible for   robots to  merge and occupy the same position (\emph{merging}). 
Once some robots occupy the same position they act as one single robot
(except possibly for a non-faulty robot merging with a faulty one).

\begin{definition}[Size-Stable Time]\
A time $t_0$ is called a \emph{size-stable time} if,
for all $t \geq t_0$,
\begin{itemize}
\item there will be no inclusions,
mergings or crossings in the system,

\item and either all agents are
inside the line segment $[x_0,x_n]$
or at most one agent is  on each side
of the line segment $[x_0,x_n]$
and stay outside of $[x_0,x_n]$.
Moreover,
the two outsiders converge to $x_0$ and $x_n$, respectively.
\end{itemize}
\end{definition}
Observe that if $t_0$ is a size-stable time,
then $t$ is a size-stable time for all $t\geq t_0$.

%
%
%

From Lemmas~\ref{lem:no-crossing} and~\ref{lem:no-more-crossing},
after a finite number of steps,
no two robots are \emph{crossing} each others.
From Lemma~\ref{lem:at most two outsiders},
either all robots are inside the line segment $[x_0,x_n]$
after a finite number of steps,
or at most two robots will stay outside of the line segment $[x_0,x_n]$ for all time $t\geq 0$.
We then get  the  following corollary.

\begin{corollary}
\label{cor:stable}
For all set of robots $X$,
there exists a size-stable time $t_0$.
\end{corollary}

Finally, from Lemmas~\ref{lem:no-crossing}, \ref{lem:no-more-crossing} and~\ref{lem:preserved-visibility},
and Corollary~\ref{cor:stable},
we can conclude that at any time after a size-stable time $t$ is reached, the farthest left and right neighbours, 
namely $l(x(t))$ and $r(x(t))$, of any robot $x$ will never change.

\begin{theorem}[Preserved-farthest-neighbours]
\label{thm:preserved-farthest}
Let $t$ be a size-stable time and $x\in {\cal R}$ be a robot.
For all $t'>t$, $r(x(t')) = r(x(t))$
and  $l(x(t')) = l(x(t))$.
\end{theorem}

For the rest of the paper,
we suppose that the earliest size-stable time is $0$.
Thus, from Corollary~\ref{cor:stable},
for all $t\geq 0$, $t$ is a size-stable time.

\subsection{Convergence of Mutual Chains}

We now define the notion of \emph{mutual chain}
as a set of robots that are mutually the farthest from each other.

\begin{definition}[Mutual Chain]
Let $0 \leq k \leq n$ be an integer and $t \geq 0$ be any size-stable time.
A \emph{mutual chain at time $t$} (or \emph{mutual chain} for short) is a configuration $C(t) = \{x_1'(t), x_2'(t), ..., x_k'(t)\} \subset X(t)$ made of $k$ robots such that
for all $1 \leq i \leq k-1$,
$l(x_{i+1}'(t)) = x_i'(t)$
and $r(x_i'(t)) = x_{i+1}'(t)$
(refer to Figure~\ref{fig:mutual-chain}).

If $r(x_i(t)) = x_j(t)$ and $l(x_j(t)) = x_i(t)$,
we say that $x_i$ and $x_j$ are
\emph{mutually chained at time $t$}
or that $x_i(t)$ and $x_j(t)$ are \emph{mutually chained}.
\end{definition}
\begin{figure}[tbh]
\centering
\includegraphics[scale = 1]{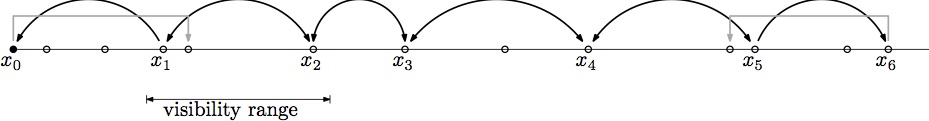}
\caption{A mutual chain of robots $C(t) = \{x_1(t),x_2(t),x_3(t),x_4(t),x_5(t)\}$
anchored in $x_0$ and $ x_6$,
where the arrows indicate farthest visibility.
\label{fig:mutual-chain}}
\end{figure}

The \emph{anchors} of a mutual chain
$C(t) = \{x_1'(t), x_2(t), ..., x_k'(t)\}$
are the farthest left neighbour of $x_1'(t)$ and the farthest right neighbour of $x_k'(t)$.
\begin{definition}[Anchors]
Given a mutual chain $C(t) = \{x_1'(t), x_2'(t), \ldots, x_k'(t)\}$,
we say that $l(x_1'(t))$ and $r(x_k'(t))$
are the left and right \emph{anchors} of $C(t)$
(or that $C(t)$ is \emph{anchored}
at $l(x_1'(t))$ and $r(x_k'(t))$)
(refer to Figure~\ref{fig:mutual-chain}).
\end{definition}

Note that the definition of anchor allows the anchors of a mutual chain
to be part of the mutual chain
(refer to Figure~\ref{fig:mutual-chain-anchors-faulty}).
\begin{figure}[tbh]
\centering
\includegraphics[scale = 1]{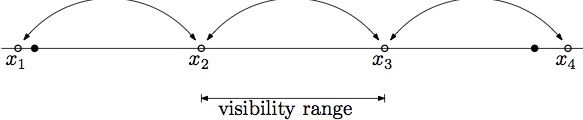}
\caption{The configuration $\{x_1,x_2,x_3,x_4\}$ is a mutual chain.
It is anchored at $x_1$ and $x_4$.
\label{fig:mutual-chain-anchors-faulty}}
\end{figure}
Moreover, the definition of mutual chain allows a mutual chain
to possibly contain only one robot
(refer to Figure~\ref{fig:mutual-chain-one-robot}).
\begin{figure}[tbh]
\centering
\includegraphics[scale = 1]{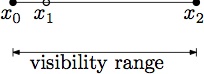}
\caption{The configuration $\{x_1\}$ is a mutual chain.
It is anchored at $x_0$ and $x_2$.
\label{fig:mutual-chain-one-robot}}
\end{figure}
Note that the anchors do not have to be faulty robots
for this situation to happen.
Indeed,
any robot $x$ forms a mutual chain $\{x(t)\}$
anchored at $l(x(t))$ and $r(x(t))$.

We now prove the formation, during the execution of the algorithm, of a special  unique mutual chain called
\emph{primary chain}. Intuitively, the primary chain is a mutual chain starting from $x_0$ and ending in $x_n$.
 We will then introduce a hierarchical   notion of  mutual chains with different levels, where
  chains of some level are  anchored in   lower level ones. Moreover, we will show that the robots will 
  eventually arrange themselves in such a hierarchical structure of mutual chains.

\begin{theorem}[Primary Chain]
\label{thm:primary}
There exists a configuration of robots
$\mathcal{C}_1 = \{x_0',x_1',x_2',...,x_k'\}\subseteq X$
such that at any size-stable time $t > 0$,
$\mathcal{C}_1(t)$ is a mutual chain anchored at
$x_0$ and $x_n$,
where $x_0'=x_0$ and $x_k'=x_n$.
This mutual chain is called the \emph{primary chain} of $X$
and it is unique.
\end{theorem}

Before we prove Theorem~\ref{thm:primary},
we need the following technical lemma.
Intuitively,
when the distance between two mutually chained robots
tends to $V$ (as $t\rightarrow\infty$),
this limit behaviour propagates to the leftmost and rightmost visible robots.
\begin{lemma}
\label{lemma:distance-propagation}
Let $x_{\alpha+1}',x_{\alpha+2}'\in X$ such that for all $t\geq 0$,
(refer to Figure~\ref{fig:propagation})
\begin{figure}[tbh]
\centering
\includegraphics[scale = 1]{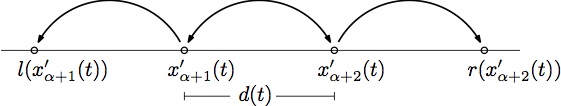}
\caption{Illustration of
Lemma~\ref{lemma:distance-propagation}.
\label{fig:propagation}}
\end{figure}
\begin{itemize}
\item $x_{\alpha+1}'(t)$ and $x_{\alpha+2}'(t)$ are mutually chained,
\item $d(t) = x_{\alpha+2}'(t) - x_{\alpha+1}'(t) \rightarrow V$,
as $t\rightarrow \infty$
\item $l(x_{\alpha+1}'(t)) \neq x_{\alpha+1}'(t)$
\item and $r(x_{\alpha+2}'(t)) \neq x_{\alpha+2}'(t)$.
\end{itemize}
Then $r(x_{\alpha+2}'(t)) - x_{\alpha+2}'(t) \rightarrow V$
and $x_{\alpha+1}'(t) - l(x_{\alpha+1}'(t)) \rightarrow V$,
as $t\rightarrow \infty$.
\end{lemma}

\proof
We have
\begin{align*}
x_{\alpha+1}'(t+1) &= l(x_{\alpha+1}'(t)) + \frac{x_{\alpha+1}'(t) - l(x_{\alpha+1}'(t))+d(t)}{2}, \\
x_{\alpha+2}'(t+1) &= x_{\alpha+1}'(t) + \frac{d(t) + r(x_{\alpha+2}'(t)) - x_{\alpha+2}'(t)}{2} .
\end{align*}
Since $x_{\alpha+1}'$ and $x_{\alpha+2}'$ are mutually chained and $d(t) \rightarrow V$ as $t\rightarrow \infty$,
there is a function $\epsilon(t)$ such that
$\epsilon(t) \rightarrow 0$ as $t\rightarrow \infty$ and
\begin{align}
\nonumber
d(t+1) &= x_{\alpha+2}'(t+1) - x_{\alpha+1}'(t+1) \\
\nonumber
&= \left(x_{\alpha+1}'(t) + \frac{d(t) + r(x_{\alpha+2}'(t)) - x_{\alpha+2}'(t)}{2}\right) - \left(l(x_{\alpha+1}'(t)) + \frac{x_{\alpha+1}'(t) - l(x_{\alpha+1}'(t))+d(t)}{2}\right)\\
\nonumber
&= \frac{x_{\alpha+1}'(t)-l(x_{\alpha+1}'(t))+r(x_{\alpha+2}'(t))-x_{\alpha+2}'(t)}{2} \\
\label{ineq:propagation}
&> V - \epsilon(t).
\end{align}
Let $\delta_1(t)$ and $\delta_2(t)$ be two functions such that
$V - \delta_1(t) = x_{\alpha+1}'(t)-l(x_{\alpha+1}'(t))$
and $V - \delta_2(t) = r(x_{\alpha+2}'(t))-x_{\alpha+2}'(t)$.
Since $l(x_{\alpha+1}'(t)) \neq x_{\alpha+1}'(t)$ and $r(x_{\alpha+2}'(t)) \neq x_{\alpha+2}'(t)$,
we have $0< \delta_1(t) \leq V$
and $0< \delta_2(t) \leq V$.
Therefore,
from~\eqref{ineq:propagation},
we get
$$\frac{V-\delta_1(t) + V-\delta_2(t)}{2} > V-\epsilon(t),$$
from which
$$0 < \frac{\delta_1(t)+\delta_2(t)}{2} < \epsilon(t) \rightarrow 0,$$
as $t\rightarrow \infty$.
This means that $\delta_1(t) \rightarrow 0$
and $\delta_2(t) \rightarrow 0$
as $t\rightarrow \infty$,
from which
$x_{\alpha+1}'(t)-l(x_{\alpha+1}'(t)) \rightarrow V$ and $r(x_{\alpha+2}'(t))-x_{\alpha+2}'(t)\rightarrow V$,
as $t\rightarrow \infty$.
\qed

\begin{proof}~(Theorem~\ref{thm:primary})

[Uniqueness]
We first explain that if the primary chain exists,
then it is unique.
Since $x_0 = x_0'$ and $x_n = x_k'$ are part of the mutual chain,
starting at $x_0$,
we get $x_1' = r(x_0)$
and $x_{i+1}' = r(x_i')$ for all $0\leq i \leq k-1$,
where $x_k' = x_n$.
So each $x_i'$ is uniquely defined.

[Existence]
We now prove that the primary chain does exist.
By Lemma~\ref{lem:at most two outsiders},
at any size-stable time $t$,
there is at most one robot $x_{\ell}$ to the left of $x_0$
which will never enter $[x_0,x_n]$
and there is at most one robot $x_r$ to the right of $x_n$
which will never enter $[x_0,x_n]$.
Moreover,
as $t \rightarrow \infty$,
$x_{\ell} \rightarrow x_0$
and $x_r \rightarrow x_n$.
Therefore,
without loss of generality,
we can ignore $x_{\ell}$ and $x_r$.
For the rest of the proof,
we suppose that $X(t) \subset [x_0,x_n]$
for any size-stable time $t$.
We need to prove that the primary chain exists.

We prove the existence of the primary chain
by contradiction.

Let us summarize the steps of the proof. We assume that there does not exist any  mutual chain. 
1) We   construct  a particular configuration,     composed by 
a forward-chain from $x_0$  connecting each node to its farthest right neighbour  till $x_n$  and a  backward chain from $x_n$
 connecting each node to its farthest left  neighbour  back to $x_0$. 
2) We then  show that the two chains converge to each other, i.e., they converge to a single chain, called {\em right-left chain}.
This construction does not directly guarantee that the right-left chain is a mutual chain.  We then show a contradiction, 
reasoning on the total length of the segment delimited by $x_0$ and $x_n$.
3) A consequence of the  right-left chain not being a mutual chain is that the total length of the segment between $x_0$ and $x_n$ is  strictly  smaller than 
$ (j+1)V $ (where $j+1$ is the number of intervals between consecutive robots in the chain).  
4)  On the other hand,  each such interval  converges to $V$, thus implying 
 that the  total length of the segment is a number arbitrarily close to $ (j+1)V $. The contradiction implies that the right-left chain is indeed mutual.

\paragraph{\bf 1) Construction of   forward and backward chains.} Let us consider a configuration of robots
$\{x_0'(t),x_1'(t),...,x_{j+1}'(t)\}\subseteq X(t)$, called {\em forward chain} (refer to Figure~\ref{fig:mutual-chain-counterexample}), 
such that:
\begin{itemize}
\item $x_0'(t) = x_0' = x_0$,

\item 
$x_{i+1}'(t) = r(x_i'(t))$ for all $0 \leq i \leq j < n$

\item and $x_{j+1}'(t) = x_{j+1}' = x_n$
\end{itemize}

\begin{figure}[tbh]
	\centering
	\includegraphics[scale = 1]{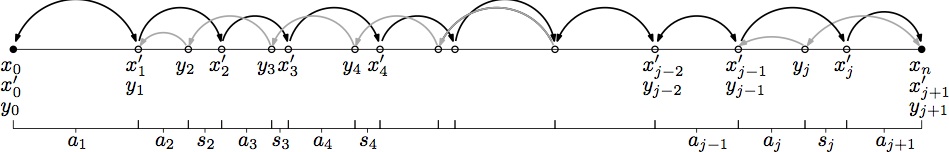}
	\caption{Illustration of the proof
of Theorem~\ref{thm:primary}.
	\label{fig:mutual-chain-counterexample}}
\end{figure}

We define another configuration of robots,  called {\em backward chain}, 
$\{y_0(t),y_1(t),...,y_{j+1}(t)\}\subseteq X(t)$
as follows.
Let $y_{j+1}(t) = x_{j+1}'(t)$
and for all $0\leq i \leq j$,
let $y_i(t) = l(y_{i+1}(t))$
(refer to Figure~\ref{fig:mutual-chain-counterexample}).
Let us call the union of the two chains {\em right-left chain}.
We now prove two useful  properties about the  right-left chain.
\begin{description}
\item[Property 1] (Alternation Property.) 
For all $1\leq i \leq j+1$,
we have $x_{i-1}'(t) < y_i(t) \leq x_i'(t)$.
We prove this property by induction,
starting at $i = j+1$.
For the base case,
notice that $y_{j+1}(t) = x_{j+1}'(t)$ by definition.
Suppose that $x_{i-1}'(t) < y_i(t) \leq x_i'(t)$
for some $1\leq i \leq j+1$.
Then,
$y_{i-1}(t) = l(y_i(t)) \leq x_{i-1}'(t)$
otherwise this would contradict the fact
that $r(x_{i-1}'(t)) = x_i'(t)$.
Moreover,
$x_{i-2}'(t) < l(y_i(t)) = y_{i-1}(t)$
otherwise this would contradict the fact
that $r(x_{i-2}'(t)) = x_{i-1}'(t)$.

\item[Property 2] (Starting point Property.) 
We have that $y_0(t) = y_0 = x_0$.
Indeed,
\begin{align*}
y_0(t) &= l(y_1(t)) && \text{by the definition of $y_0(t)$,}\\
&\leq l(x_1'(t)) && \text{by Property 1,}\\
&\leq x_0,
\end{align*}
otherwise $x_1'(t)$ would not be the rightmost visible robot from $x_0 = x_0'$.
\end{description}

\paragraph{\bf 2) Convergence of forward and backward chains to a  right-left chain.}  Notice that since the forward chain  $\{x_0'(t),x_1'(t),...,x_{j+1}'(t)\}$
is not a mutual chain,
there exists an $i$ with $1 \leq i \leq j$
such that $x_{i-1}'(t) < y_i(t) < x_i'(t)$.
For all $1 \leq i \leq j+1$,
let $a_i(t) = y_i(t) - x_{i-1}'(t)$
and $s_i(t) = x_i'(t) - y_i(t)$.
Our aim, in the following, is to prove that  $x_i'(t)$ and $y_i(t)$ get arbitrarily close
  for $t$ going to infinity.
  
From Property 1,
we have $a_i(t) > 0$ and $s_i(t) \geq 0$
for all $1 \leq i \leq j+1$.
Moreover,
$s_i(t) = 0$
if and only if
$y_i(t) = x_i'(t)$.
Notice that $l(x_i'(t-1)) \leq x_{i-1}'(t-1)$,
otherwise there would be a contradiction with  the fact that
$r(x_{i-1}'(t-1)) = x_i'(t-1)$.
Therefore,
$$x_i'(t) = \frac{l(x_i'(t-1)) + r(x_i'(t-1))}{2} \leq \frac{x_{i-1}'(t-1) + x_{i+1}'(t-1)}{2} ,$$
from which
\begin{align}
\label{equationX}
x_i'(t) &\leq
\begin{cases}
0 & i=0 ,\cr
x_{i-1}'(t-1) + \frac{1}{2}(a_i(t-1)+s_i(t-1)+a_{i+1}(t-1)+s_{i+1}(t-1)) & 1 \leq i \leq j , \cr
x_n & i = j+1 . \cr
\end{cases}
\end{align}
Moreover,
notice that $r(y_i(t-1)) \geq y_{i+1}(t-1)$,
otherwise there would be a contradiction with  the fact that
$l(y_{i+1}(t-1)) = y_i(t-1)$.
Therefore,
$$y_i(t) = \frac{l(y_i(t-1)) + r(y_i(t-1))}{2} \geq \frac{y_{i-1}(t-1) + y_{i+1}(t-1)}{2} ,$$
from which
\begin{align}
\label{equationY}
y_i(t) &\geq
\begin{cases}
0 & i = 0 ,\cr
y_{i-1}(t-1) +\frac{1}{2}(s_{i-1}(t-1) + a_i(t-1)+s_i(t-1)+a_{i+1}(t-1)) & 1 \leq i \leq j ,\cr
x_n & i = j+1 .\cr
\end{cases}
\end{align}

Since $s_i(t) = x_i'(t) - y_i(t) $,
by subtracting~\eqref{equationY} from~\eqref{equationX} we obtain
\begin{align}
\label{equationSi}
s_i(t) &\leq
\begin{cases}
0 & i=0 , \cr
\frac{1}{2}(s_{i-1}(t-1) + s_{i+1}(t-1)) & 1 \leq i \leq j , \cr
0 & i = j+1 . \cr
\end{cases}
\end{align}
We are now ready to prove that   for all $0 \leq i \leq j+1$,
$s_i(t) \rightarrow 0$ as $t\rightarrow \infty$, implying that
   $y_i(t) \rightarrow x_i'(t)$ as $t\rightarrow \infty$.
Notice that we already have $y_0(t) = x_0'(t)$
and $y_{j+1}(t) = x_{j+1}'(t)$ by definition.
We then have:
\begin{align*}
s_i(t) &\leq \frac{1}{2}(s_{i-1}(t-1) + s_{i+1}(t-1)) \\
&\leq \frac{1}{4}(s_{i-2}(t-2) + 2s_{i}(t-2) + s_{i+2}(t-2)) \\
&\leq \frac{1}{8}(s_{i-3}(t-3) + 3 s_{i-1}(t-3) + 3 s_{i+1}(t-3) + s_{i+3}(t-3)) \\
&\leq \frac{1}{16}(s_{i-4}(t-4) +4 s_{i-2}(t-4) + 6 s_i(t-4) + 4s_{i+2}(t-4) + s_{i+4}(t-4)) \\
&\vdots \\
&\leq \frac{1}{2^t}\sum_{k = 0}^t \binom{t}{k} s_{i-t+2k}(0),
\end{align*}
where $s_i(t) = 0$ for all $i \leq 0$ and $i \geq j+1$.

In order to determine the limit of $s_i(t)$
when $t \rightarrow \infty$,
we need to make a few observations.
First of all, the $s_i(t)$'s in the summation with 
$i \leq 0$ or $i \geq j+1$
are all equal to zero.
In other words, regardless of the value of $t$,
there are at most $j$ non-zero values in the summation.
These $j$ values correspond to the $j$-central binomial coefficients.
Also note that since the segment delimited by the two faulty robots has a constant size, the values of the $s_i$'s are bounded. Let $C$ be the value of the largest such $s_i$ ever occurring. Since the largest binomial coefficient is the central one (or the central ones for odd values of $t$),
we can write
$$
0 \leq s_i(t) \leq \frac{1}{2^t} \, j \,
\binom{t}{\lfloor\frac{t}{2}\rfloor} \, C
.
$$
Since\footnote{We write $f(t) \sim g(t)$ whenever
$\lim_{t\rightarrow\infty} \frac{f(t)}{g(t)} = 1$.}
$\binom{t}{\lfloor \frac{t}{2} \rfloor} \sim \frac{2^t}{\sqrt{ \pi \frac{t}{2}}}$,
we have
$$0 \leq \lim_{t \rightarrow \infty} s_i(t)
\leq \lim_{t \rightarrow \infty} \frac{1}{2^t} \, j \,
\binom{t}{\lfloor\frac{t}{2}\rfloor} \, C = \lim_{t \rightarrow \infty} \frac{1}{2^t} \, j \,
\frac{2^t}{\sqrt{ \pi \frac{t}{2}}} \, C
= 0,$$
from which $\lim_{t \rightarrow \infty} s_i(t) = 0$.

We are ready to derive a contradiction.

\paragraph{\bf 3) Length of the segment strictly smaller than   $(j+1)V$.} 
Since the right-left chain is not a mutual chain,
and $x_0$ and $x_n$ are not moving,
the distance between $x_0$ and $x_n$ must be strictly smaller than $(j+1)V$ (otherwise
$x'_{j}$ and $y_j$ would necessarily coincide, for all $j$).
So, there exists a real number $\delta > 0$ such that
$x_n - x_0 = (j+1)V - \delta$.


\paragraph{\bf  4) Distance between $x'_i(t)$ and  $x'_{i+1}(t)$ tending to $V$.} 
Let us consider any sub-chain of the  right-left chain for which the $x'_i$ and the $y_i$ are distinct except for the extremal ones.
More precisely, let $\alpha$ and $\beta$ be two indices such that
$x_{\alpha}' = y_{\alpha}$, $x_{\beta}' = y_{\beta}$
and $x_i' \neq y_i$ for all $\alpha < i < \beta$
(refer to Figure~\ref{fig:mutual-chain-propagation-V}).
\begin{figure}[tbh]
	\centering
	\includegraphics[scale = 1]{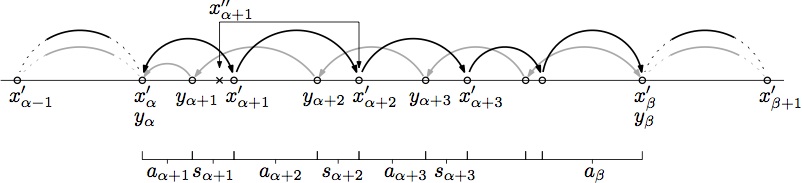}
	\caption{Illustration of the contradiction in the proof
of Theorem~\ref{thm:primary} (propagation of distance $V$).
We do not make any assumption about $x_{\alpha-1}'$
being equal or not to $y_{\alpha-1}$,
nor about $x_{\beta+1}'$
being equal or not to $y_{\beta+1}$.
	\label{fig:mutual-chain-propagation-V}}
\end{figure}
Notice that
$l(x_{\alpha+1}') = x_{\alpha}'$,
otherwise this would contradict the fact that
$l(y_{\alpha+1}) = x_{\alpha}'$.
We also have
$r(y_{\beta-1}) = x_{\beta}'$,
otherwise
this would contradict the fact that
$r(x_{\beta-1}') = x_{\beta}'$.
Therefore,
$l(x_{\alpha+1}') = x_{\alpha}'$, $r(x_{\alpha+1}') = x_{\alpha+2}'$,
$l(y_{\beta-1}) = y_{\beta-2}$ and $r(y_{\beta-1}) = y_{\beta} = x_{\beta}'$.
This implies that $k\geq i+3$,
otherwise $x_{\alpha+1}'$ and $y_{\beta-1}$
would have the same leftmost and rightmost
visible robots and they would merge in one step,
which is not possible at a size-stable time.
Since there cannot be any merging,
given that $l(y_{\alpha+1}) = y_{\alpha} = x_{\alpha}'$,
we must also have that
$x_{\alpha+2}'$ is not visible from $y_{\alpha+1}$ at any time.
Therefore,
for all $t\geq 0$,
$s_{\alpha+1}(t) + a_{\alpha+2}(t) + s_{\alpha+2}(t) > V$.
Since $r(x_{\alpha+1}') = x_{\alpha+2}'$,
for all $t\geq 0$,
$a_{\alpha+2}(t) + s_{\alpha+2}(t) \leq V$.
Together with the fact that
$s_{\alpha+1}(t) \rightarrow 0$
and $s_{\alpha+2}(t) \rightarrow 0$
as $t \rightarrow \infty$,
we get that $a_{\alpha+2}(t) \rightarrow V$
as $t \rightarrow \infty$.
Therefore,
$x_{\alpha+2}'(t) - x_{\alpha+1}'(t) \rightarrow V$
as $t \rightarrow \infty$.

Our goal is to apply Lemma~\ref{lemma:distance-propagation}
and conclude that
$x_{\alpha+1}'(t) - x_{\alpha}' \rightarrow V$
and $x_{\alpha+3}'(t) - x_{\alpha+2}' \rightarrow V$
as $t \rightarrow \infty$.
However,
since $x_{\alpha+1}'(t)$ and $x_{\alpha+2}'(t)$
are not mutual,
we cannot apply the lemma directly.
Here is the idea we use to circumvent this problem.
We can prove that there is a robot $x_{\alpha+1}''(t)$,
satisfying
$y_{\alpha+1}(t) \leq x_{\alpha+1}''(t) \leq x_{\alpha+1}'(t)$,
that is mutually chained with $x_{\alpha+2}'(t)$.
Intuitively,
since $y_{\alpha+1}(t)\rightarrow x_{\alpha+1}'(t)$
as $t\rightarrow\infty$,
and since $x_{\alpha+1}''(t) \in [y_{\alpha+1}(t), x_{\alpha+1}'(t)]$,
$x_{\alpha+1}''$ behaves the same way $x_{\alpha+1}'$ does.
But since $x_{\alpha+1}''(t)$ is mutually chained with $x_{\alpha+2}'(t)$,
we can apply Lemma~\ref{lemma:distance-propagation}.
Formally, let
$x_{\alpha+1}''(t) = l(x_{\alpha+2}'(t))$.
Notice that we must  have the following:
$x_{\alpha+1}''(t) \leq x_{\alpha+1}'(t)$, $y_{\alpha+1}(t) \leq x_{\alpha+1}''(t)$, $r(x_{\alpha+1}''(t)) = x_{\alpha+2}'(t)$,
otherwise we would have  contradictions, respectively, with the following three facts:
$r(x_{\alpha+1}'(t)) = x_{\alpha+2}'(t)$, $l(y_{\alpha+2}(t)) = y_{\alpha+1}(t)$, and $r(x_{\alpha+1}'(t)) = x_{\alpha+2}'(t)$.
Since $y_{\alpha+1}(t) \rightarrow x_{\alpha+1}'(t)$
as $t\rightarrow \infty$,
then $x_{\alpha+1}''(t) \rightarrow x_{\alpha+1}'(t)$
as $t\rightarrow \infty$.
The fact that
$x_{\alpha+2}'(t) - x_{\alpha+1}' \rightarrow V$
as $t \rightarrow \infty$
therefore implies that
$x_{\alpha+2}'(t) - x_{\alpha+1}''(t) \rightarrow V$
as $t \rightarrow \infty$.
By Lemma~\ref{lemma:distance-propagation},
$x_{\alpha+1}'(t) - x_{\alpha}' \rightarrow V$
and $x_{\alpha+3}'(t) - x_{\alpha+2}' \rightarrow V$
as $t \rightarrow \infty$.

By the previous argument,
the fact that
$x_{\alpha+2}'(t) - x_{\alpha+1}' \rightarrow V$
as $t\rightarrow \infty$
propagates to $x_{\alpha+1}'(t) - x_{\alpha}'$
and $x_{\alpha+3}'(t) - x_{\alpha+2}'$.
We can repeat the same argument and show that
this propagates to all $x_i'$'s,
from which we get that
for all $0 \leq i \leq j$,
$x_{i+1}'(t) - x_i' \rightarrow V$
as $t\rightarrow \infty$.
Therefore,
the total distance between $x_0$ and $x_n$ is arbitrarily close to $(j+1)V$.
This contradicts the fact that
$x_n(t) - x_0(t) = (j+1)V - \delta$
for all $t\geq 0$.
\end{proof}

In the proof of Theorem~\ref{thm:primary},
we showed the existence of a unique mutual chain
called the primary chain.
Intuitively,
we say that a configuration of robots is a secondary chain
if it is a mutual chain anchored at two robots
that belong to the primary chain.
However,
such a configuration is not necessary unique
(refer to Figure~\ref{fig:example-secondary}
for an example).
\begin{figure}[tbh]
\centering
\includegraphics[scale = 1]{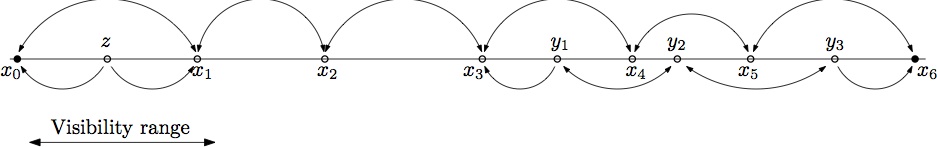}
\caption{An example of a primary chain
$\{x_0, x_1, \ldots x_6\}$
with two level-$2$ chains:
$\{z\}$ (anchored at $x_0$ and $x_1$)
and $\{y_1,y_2,y_3\}$ (anchored at $x_3$ and $x_6$).}
\label{fig:example-secondary}
\end{figure}
Level-$j$ chains (for $j>2$)
are defined in a similar way.
\begin{definition}[Secondary Chains and Level-$j$ Chains]\
\label{def-level-j}
\begin{itemize}
\item The primary chain is called a \emph{level-$1$ chain}.

\item A configuration of robots $C$ is a \emph{secondary chain} if it is a mutual chain anchored at two robots $x$ and $x'$,
such that 
$x,x' \in \mathcal{C}_1$ and and least one of $x$ and $x'$ is non-faulty.
We say that a secondary chain is a \emph{level-$2$ chain}.

\item A configuration of robots $C$ is a \emph{level-$j$ chain} if it is a mutual chain anchored at two robots $x$ and $x'$ which satisfy the following property.
There exists an index $j' < j$ such that one of the following two statements is true:
\begin{itemize}
\item $x$ is part of a level-$j'$ chain and $x'$ is part of a level-$(j-1)$ chain

\item or $x$ is part of a level-$(j-1)$ chain and $x'$ is part of a level-$j'$ chain.
\end{itemize}
\end{itemize}
\end{definition}

The convergence of the primary chain can be proven by observing
that the behaviour of the robots in the primary chain executing
 our algorithm ({\sc Convergence1D})
is    equivalent to the behavior
 they would have if they were executing  Algorithm  {\sc Spreading}. Once this is established, convergence follows from  
  Theorem  \ref{thm:generalizedCohen}.
%
%
The following lemma shows under what conditions Theorem  \ref{thm:generalizedCohen}   can be applied to a general mutual chain $Y(t) = \{y_1,y_2,\ldots, y_k\}$.
More specifically,
suppose that there exists two real numbers $y_0'$ and $y_{k+1}'$ such that $y_0(t) = l(y_1(t)) \rightarrow y_0'(t)$ and $y_{k+1}(t) = r(y_k(t)) \rightarrow y_{k+1}'$
as $t\rightarrow \infty$.
Then, by applying Algorithm Convergence1D,
$Y(t)$ converges towards an equidistant
configuration between $y_0'(t)$ and $y_{k+1}'(t)$.
\begin{lemma}
\label{lem:condi}
Let $Y(t) = \{y_1(t),y_2(t),\ldots, y_k(t)\}$ be a mutual chain at a size-stable time $t$,
anchored in $y_0(t)=l(y_1(t))$ and $y_{k+1}(t)=r(y_k(t))$,
where $y_0(t) \neq y_1(t)$ and $y_{k+1}(t) \neq y_k(t)$.
Suppose that there exist two numbers $y_0'$ and $y_{k+1}'$,
such that $y_0(t) \rightarrow y_0'$
and $y_{k+1}(t) \rightarrow y_{k+1}'$
as $t\rightarrow \infty$.
We have that,
for all $0\leq i \leq k+1$,
$$y_i(t) \rightarrow y_0' + \frac{|y_{k+1}'-y_0'|}{k+1}\,i$$
as $t\rightarrow\infty$. Therefore,
as $t\rightarrow\infty$, the robots in $\{y_1(t),y_2(t),\ldots, y_k(t)\}$ converge to a configuration where the distance between any two consecutive robots
is $\frac{|y_{k+1}'-y_0'|}{k+1}$.
\end{lemma}

\begin{proof}
Let $Z(t) = \{z_0(t)=y_0(t), z_1(t),z_2(t),\ldots, z_m(t)= y_{k+1}(t)\}$ be the global configuration of robots
at time $t$,
restricted to the interval $[y_0(t),y_{k+1}(t)]$.

By Theorem~\ref{thm:preserved-farthest},
$Y(t)$ satisfies the following property:
for all $1\leq i \leq k$
and for all $t'\geq t$,
$l(y_i(t')) = l(y_i(t))$
and $r(y_i(t')) = r(y_i(t))$.
Therefore,
even if there is a robot $z_j(t) \in N(y_i(t))\setminus Y(t)$,
the  presence of  $z_j(t)$  has no impact on the position of $y_i(t+1)$.
Consequently,
the positions of the robots in $Y(t+1)$,
after executing Algorithm {\sc Convergence1D} on $Y(t)$,
are uniquely determined by the positions
of the robots in $Y(t)$.
Hence,
executing Algorithm {\sc Convergence1D} on $Y(t)$
produces the same result as
executing Algorithm {\sc Spreading} on $Y(t)$, and thus  the lemma follows from Theorem~\ref{thm:generalizedCohen}.
\end{proof}

We now show that the primary chain
$\mathcal{C}_1 = \{x_0',x_1',x_2',...,x_k'\}\subseteq X$,
where $x_0' = x_0$ and $x_k' = x_n$,
converges towards a configuration of equidistant robots delimited by its anchors $x_0$ and $x_n$.
\begin{theorem}[Convergence of the Primary Chain]\
\label{thm:PrimaryConvergence}
Let $\mathcal{C}_1 = \{x_0',x_1',x_2',...,x_k'\}$
be the primary chain.
We have that $x_0' = x_0$, $x_k' = x_n$
and for all $0 \leq i \leq k$
$$x_i'(t) \rightarrow \frac{|x_n-x_0|}{k}\,i$$
as $t\rightarrow \infty$.
\end{theorem}

\begin{proof}
Since $\mathcal{C}_1$ is a mutual chain,
the configuration $\{x_1',x_2',...,x_{k-1}'\}$ is also a mutual chain. It is anchored at $x_0'$ and $x_k'$,
where $x_0' \neq x_1'$ and $x_k' \neq x_{k-1}'$.
Since the anchors $x_0' = x_0 = 0$ and $x_k' = x_n$
are faulty,
they do not move.
Hence,
$x_0'(t) \rightarrow x_0$
and $x_k'(t) \rightarrow x_n$
as $t\rightarrow \infty$.
Thus,
the theorem follows directly from
Lemma~\ref{lem:condi}.
\end{proof}

We now show that every level-$j$ chain converges towards
a configuration of equidistant robots. 
\begin{theorem}[Convergence of Level-$j$ Chains]\
\label{thm:level-j}
Let $C_j = \{y_1,y_2,\ldots, y_k\}$ be a level-$j$ chain,
where $j\geq 1$ is an integer.
Let $t$ be a size-stable time.
Let $y_0(t) = l(y_1(t))$ and $y_{k+1}(t) = r(y_k(t))$.
There exist real numbers $y_0'$ and $y_{k+1}'$
such that $y_0(t) \rightarrow y_0'$
and $y_{k+1}(t) \rightarrow y_{k+1}'$
as $t\rightarrow \infty$.
Moreover,
for all $0 \leq i \leq k+1$,
$$y_i(t) \rightarrow y_0' + \frac{|y_{k+1}'-y_0'|}{k+1}\,i$$
as $t\rightarrow\infty$.
\end{theorem}

\begin{proof}
We proceed by induction on $j$.
By Theorem~\ref{thm:PrimaryConvergence},
our statement  is true for $j=1$.
Suppose that the theorem is true for all integers from $1$ to $j-1$. Consider a level-$j$ chain
$C_j = \{y_1,y_2,\ldots, y_k\}$
anchored at $y_0(t) = l(y_1(t))$
and $y_{k+1}(t) = r(y_k(t))$,
where $t$ is a size-stable time.

By Defintion~\ref{def-level-j},
there exists an index $j' < j$ such that one of the following two statements is true:
\begin{itemize}
\item $y_0$ is part of a level-$j'$ chain and $y_{k+1}$ is part of a level-$(j-1)$ chain

\item or $y_0$ is part of a level-$(j-1)$ chain and $y_{k+1}$ is part of a level-$j'$ chain.
\end{itemize}
Without loss of generality,
suppose that $y_0$ is part of a level-$j'$ chain and $y_{k+1}$ is part of a level-$(j-1)$ chain.

By the induction hypothesis,
there exist two real numbers $y_0'$ and $y_{k+1}'$
such that $y_0(t) \rightarrow y_0'$
and $y_{k+1}(t) \rightarrow y_{k+1}$
as $t\rightarrow \infty$.
The theorem follows from
Lemma~\ref{lem:condi}.
\end{proof}

The following lemma states that every robot belongs to some
level-$j$ chain.
To simplify the presentation,
we assume that the faulty robot $x_0$ is part of the \emph{level-$0$} chain $\{x_0\}$
and that the faulty robot $x_n$ is part of the \emph{level-$0$} chain $\{x_n\}$.
\begin{lemma}
\label{lem:any}
For all size-stable time $t$
and all $0 \leq i \leq n$,
$x_i(t) \in X(t)$ belongs to a level-$j$ chain.
\end{lemma}

\begin{proof}
Suppose that the statement is false.
Let $y_1(t)$ be the leftmost robot
that does not satisfy the statement.
We will derive a contradiction.

Since the leftmost robot $x_0$ is faulty,
$l(y_1(t))$ belongs to a mutual chain,
say $C(t) = \{x_1'',x_2'',...,x_m''\}$,
where $l(y_1(t))=x_{\alpha}''$
for some index $1\leq \alpha\leq m$.
Let $Y = \{y_1,y_2,...,y_k\}$ be the configuration of robots
such that
(refer to Figure~\ref{fig:every-robot-belongs})
\begin{figure}[tbh]
\centering
\includegraphics[scale = 1]{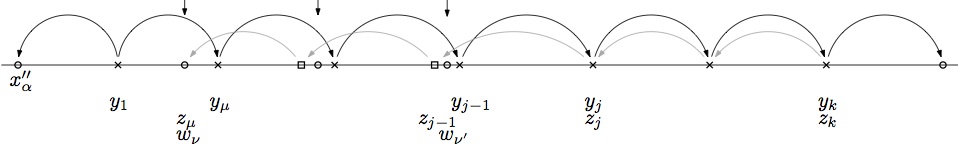}
\caption{Illustration of the proof of Lemma~\ref{lem:any}.
\label{fig:every-robot-belongs}}
\end{figure}
\begin{enumerate}
\item $y_i(t) = r(y_{i-1}(t))$ for all $2\leq i \leq k$,

\item\label{EveryRobotBelongs-item2}
$r(y_k(t))$ belongs to a mutual chain

\item\label{EveryRobotBelongs-item3}
and for all $1\leq i \leq k$, $y_i(t)$
does not belong to a mutual chain.
\end{enumerate}
Observe that the definition of $Y$
allows $k$ to be equal to $1$
(in such a case,
only items~\ref{EveryRobotBelongs-item2}
and~\ref{EveryRobotBelongs-item3} apply).
By construction and the definition of $y_1(t)$,$\{y_1(t),y_2(t),...,y_k(t)\}$ is not a mutual chain.
Therefore,
for the rest of the proof,
$k \geq 2$.

Let $\{z_1,z_2,...,z_k\}$
be the configuration of robots such that
$z_k = y_k$ and
$z_i(t) = l(z_{i+1}(t))$ for all $1\leq i\leq k-1$.
Using the same arguments as in
the proof of Theorem~\ref{thm:primary},
we get that $x_{\alpha}'' \leq z_1 \leq y_1$
and $y_{i-1} < z_i \leq y_i$
for all $2 \leq i \leq k$.
Since $\{y_1(t),y_2(t),...,y_k(t)\}$
is not a mutual chain,
there is an index $i$ such that $z_i(t) \neq y_i(t)$.
Let $j$ be the smallest index such that
$z_j = y_j$ and $z_{j-1} \neq y_{j-1}$.
Suppose there is an index $\gamma < j-1$
such that $z_{\gamma}(t) = y_{\gamma}(t)$.
Therefore, by the definition of $j$,
$z_i = y_i$ for all $1\leq i \leq \gamma$.
Moreover,
$x_{\alpha}''$ and $r(y_k)$ are part of mutual chains.
Therefore,
by Theorems~\ref{thm:PrimaryConvergence}
and~\ref{thm:level-j},
$x_{\alpha}''(t)$ and $r(y_k)(t)$ converge to a fixed location as $t\rightarrow \infty$.
Consequently,
we get the same contradiction as in the proof
of Theorem~\ref{thm:primary}.
Hence,
for the rest of the proof,
assume that $z_i(t) \neq y_i(t)$
for all $1 \leq i < j-1$.

We have the following property.
\begin{description}
\item[Property 1]
If,
for all $2\leq i \leq j-1$,
$z_i(t)$ does not belong to any mutual chain,
then $z_1(t) = l(z_2(t))$ belongs to a mutual chain.
Indeed,
we must have $z_1(t) \leq y_1(t)$ otherwise this would contradict the fact that $y_2(t) = r(y_1(t))$.
Moreover,
we assumed that $z_i(t) \neq y_i(t)$
for all $1 \leq i < j-1$.
Hence,
$z_1(t) < y_1(t)$.
Moreover,
we must have $z_1(t) \geq x_{\alpha}''(t)$ otherwise this would contradict the fact that $x_{\alpha}'' = l(y_1(t))$.
But then,
since $y_1(t)$ is the leftmost robot 
that does not belong to a mutual chain,
we must have that $z_1(t) = l(z_2(t))$
belongs to a mutual chain.
\end{description}
Consequently,
there is an index $1\leq i \leq j-1$
such that $z_i$ belongs to a mutual chain.
Let $1\leq \mu \leq j-1$ be the largest index
such that $z_{\mu}$ belongs to a mutual chain,
say $W = \{w_1,w_2,...,w_{m'}\}$.
Let $1 \leq \nu \leq m'$ be the index such that
$w_{\nu} = z_{\mu}$.

We have the following property.
\begin{description}
\item[Property 2]
$z_{\mu+1} < w_{\nu+1} < y_{\mu+1}$.
Indeed,
observe that $w_{\nu+1} = r(w_{\nu})$.
Therefore,
$w_{\nu+1} \leq y_{\mu+1}$
otherwise this would contradict the fact that
$y_{\mu+1} = r(y_{\mu})$.
Moreover,
by definition,
$w_{\nu+1} \neq y_{\mu+1}$.
We also have that
$w_{\nu+1} \geq z_{\mu+1}$
otherwise this would contradict the fact that
$z_{\mu} = l(z_{\mu+1})$.
Moreover,
by definition,
$w_{\nu+1} \neq z_{\mu+1}$.
\end{description}
By repeating the argument for proving Property 2,
we reach the index $\nu'$ such that
$z_{j-1} < w_{\nu'}' < y_{j-1}$.
Observe that $w_{\nu'+1} = r(w_{\nu'}) \leq y_j$
otherwise this would contradict the fact
$y_j = r(y_{j-1})$.
Moreover,
$w_{\nu'+1} \geq y_j = z_j$
otherwise this would contradict the fact
$z_{j-1} = l(z_j)$.
Therefore,
$w_{\nu'+1} = y_j$.
However,
by the definition of $Y$,
$y_j$ is not part of a mutual chain.
We get a contradiction.
\end{proof}

From Theorems~\ref{thm:PrimaryConvergence}
and~\ref{thm:level-j},
and Lemma~\ref{lem:any},
we can conclude with the following theorem. 
\begin{theorem}[Global Convergence]\
\label{thm:convergence}
For all $0\leq i \leq n$,
$|x_i(t+1) - x_i(t)| \rightarrow 0$
as $t\rightarrow \infty$.
Therefore,
$X(t)$ converges towards a fixed configuration
as $t\rightarrow \infty$.
\end{theorem}

%
%
%
%

%% file: Conclusion.tex

To study the impact of faults on the robots dynamics, in this paper we analyzed the behaviour of a group of oblivious robots 
which execute an algorithm designed for a fault-free environment in presence of undetectable crash faults. 
We focused on  the classic  point-convergence algorithm by Ando et al. \cite{AnOaSuYa99} executed on a line, when the robots are  synchronous and at most 
two of them are faulty.
 
 The paper leaves several open questions and  research directions. 
An obvious extension would be 
the study of  the  point-convergence algorithm   in  the case of more than  two faults: we know that the robots  still converge to a pattern, but the analysis is not simple and  left for further study. 
 %
When the robots operate fully synchronously in a  two dimensional space, the dynamics has a rather different nature: we have observed that
 oscillations are possible, even with just two faults and the study of this case  is undergoing.

More generally, this work can be seen as a first step toward the study  of 
 the interaction between heterogeneous groups of robots operating in the same space, each   following a different algorithm. 
The existing literature on {\sc Look}-{\sc Compute}-{\sc Move} robots has always  considered 
robots 
with the same set of rules.  The presence of different teams following different,   
possibly conflicting, rules  in the environment is an interesting new area of investigation.